\documentclass[times,10pt]{article}
\usepackage{fullpage}
\usepackage{latex8}
\usepackage{blkarray}
\usepackage{bbm}
\usepackage{graphicx}
\usepackage{graphics}
\usepackage[hang]{subfigure}
\usepackage{amssymb}
\usepackage{epsfig}
\usepackage{epic,eepic}
\usepackage{tikz}

\usetikzlibrary{trees}
\usetikzlibrary{arrows,automata}

\newtheorem{theorem}{{\bf Theorem}}

\def\QED{\mbox{\rule[0pt]{1.5ex}{1.5ex}} \vspace{0.2cm}}
\def\cqfd{\hspace*{\fill}~\QED\par\endtrivlist\unskip}

\newenvironment{theorem-repeat}[1]{\begin{trivlist}
\item[\hspace{\labelsep}{\bf\noindent Theorem~\ref{#1} }]}%
{\end{trivlist}}
\newcommand{\toto}{xxx}

\newenvironment{lemma-repeat}[1]{\begin{trivlist}
\item[\hspace{\labelsep}{\bf\noindent Lemma~\ref{#1} }]}%
{\end{trivlist}}
\newcommand{\titi}{xxx}

\newenvironment{proposition-repeat}[1]{\begin{trivlist}
\item[\hspace{\labelsep}{\bf\noindent Lemma~\ref{#1} }]}%
{\end{trivlist}}
\newcommand{\tutu}{xxx}

\def\QED{\mbox{\rule[0pt]{1.5ex}{1.5ex}} \vspace{0.2cm}}
\def\cqfd{\hspace*{\fill}~\QED\par\endtrivlist\unskip}

\begin{document}

\newcommand{\DS}[1]{{\displaystyle #1}}
\newcommand{\eps}{\varepsilon}
\newcommand{\UL}[1]{{\underline #1}}
\newcommand{\n}{\underline{n}}
\newcommand{\e}{\nbu_{k,m}}
\newcommand{\nbR}{\mathbbm{R}}
\newcommand{\nbN}{\mathbbm{N}}
\newcommand{\nbQ}{\mathbbm{Q}}
\newcommand{\nbC}{\mathbbm{C}}
\newcommand{\nbu}{\mathbbm{1}}
\newcommand{\nbI}{\mathbbm{I}}
\newcommand{\nbP}{\mathbbm{P}}
\newcommand{\nbZ}{\mathbbm{Z}}
\newcommand{\nbK}{\mathbbm{K}}
\newcommand{\nbE}{\mathbbm{E}}
\newcommand{\nbun}{1}
\newcommand{\ten}{\longrightarrow}
\newcommand{\s}{{\cal S}}
\renewcommand{\Pr}{\nbP}

\newcommand{\thetathank}{\thanks{Supported by
     the Direction G\'en\'erale des Entreprises - P2Pim@ges project}}



\title{Analytical Study of Adversarial Strategies in Cluster-based Overlays}

\author{E. Anceaume \\
CNRS/IRISA,  France\\
 \and  R. Ludinard,  B. Sericola\\
INRIA Rennes Bretagne-Atlantique, France\\
\and F. Tronel\\
Supelec, France\\
\and F. Brasiliero\\
Universidade Federal de Campina Grande, LSD Laboratory, Brazil
}

\maketitle

\begin{abstract}

Scheideler  has  shown that peer-to-peer overlays networks can only survive Byzantine attacks if malicious nodes are  not able to predict what is going to be the topology of the network for a given sequence of join and leave operations.   In this paper we investigate adversarial strategies by  following specific games. Our analysis  demonstrates  first that
 an adversary can very quickly  subvert DHT-based overlays by simply never triggering leave operations. We then show  that  when all nodes (honest and malicious ones) are imposed on a limited lifetime, the system eventually reaches a stationary regime where the ratio of polluted  clusters  is bounded, independently from the initial amount of corruption in the system.
 \end{abstract}

\section{Introduction}

The adoption of peer-to-peer overlay networks as a building block for
architecting Internet scale systems has raised the attention of making
these overlays resilient not only to benign crashes, but also to more
malicious failure models for the peers~\cite{CDGRW02,singhinfocom06,sitmorris02,SL04}. 
As a result, Byzantine-resilient overlay systems have been
proposed (e.g.,~\cite{robustchord2005,BaumgartM07,ABLR08}). The key to achieve 
Byzantine resilience in a peer-to-peer overlay is
to prevent malicious peers from isolating correct ones. This in turn,
can only be achieved if malicious peers are not able to predict what
will be the topology of the overlay for a given sequence of join and
leave operations. Hence, a prerequisite for this condition to hold is to guarantee 
that malicious nodes are well-mixed with honest ones, that is nodes identifiers randomness 
 is continuously preserved. Unfortunately, targeted join/leave attacks 
may quickly endanger the relevance of such
assumption. Actually by holding a logarithmic number of IP addresses, an adversary can very easily and efficiently disconnect some target from the rest of the system. This can be achieved  in a linear number of offline trials~\cite{AS04}. 
Awerbuch and Scheideler~\cite{AS07}  have analysed 
several ways to make overlay networks provably 
robust against
different forms of malicious attacks, and in particular targeted join/leave attacks, through competitive
 algorithms. All these solutions are based on
the introduction of locally induced churn to prevent the adversary from
thwarting randomness. The same authors have shown that 
despite the high level of randomness introduced in each of these strategies, most of them are either incorrect, 
 or  they 
  involve tight synchronization among nodes which becomes unbearable in the context we address, namely targeted and frequent join/leave attacks. The other proposed approach based on globally induced churn, enforce limited lifetime for
  each node in the system. However, these solutions  keep the system in an unnecessary  hyper-activity, and thus need to impose strict
restrictions on  nodes joining rate which
clearly limit their applicability to open systems.

In this paper we propose to leverage the power of clustering to design a  practically usable solution that preserves randomness under an $\epsilon$-bounded adversary. 
Our solution relies on the clusterized version of peer-to-peer overlays combined with a mechanism that allows the enforcement of  limited nodes lifetime. Clusterized versions of structured-based overlays are such that  clusters of nodes substitute  nodes at the vertices of the graph. 
Cluster-based overlays have revealed to be well adapted for efficiently reducing the impact of churn on the system and/ or in greatly reducing the damage caused by failures---assuming that  failures assumptions hold anywhere and at any time in the system~\cite{ABLR08,equus,robustchord2005}. 

The contributions of the paper are two-fold. First we investigate adversarial strategies by  following specific games. Our analysis  demonstrates  that
 an adversary can very quickly  subvert cluster-based overlays by simply never triggering leave operations. We then show  that  when nodes are imposed on a limited lifetime and under the assumption that we are able to enforce the adversary to leave the system after expiration of its ID,  the system eventually reaches a stationary regime where the ratio of polluted  clusters  is bounded. Second we propose a simple and generic  mechanism to limit nodes  lifetime in those systems.

The remainder of this paper is as follows:  In Section~\ref{sec:peercube} we briefly describe the main features of cluster-based overlays, and propose a mechanism that enables the enforcement of limited nodes lifetime. In Section~\ref{sec:model}, we model adversarial behaviours through the use of games. We study the outcome of these games by using a Markovian analysis. In this section, we consider a non restricted adversary. Section~\ref{sec:constraint-model} is devoted to the same study in the case of a restricted adversary. Finally, we conclude with future  works.

\section{Cluster-based DHT Overlays  in a Nutshell}
\label{sec:peercube}
In this section we first present the common features of cluster-based overlays and then present different join/leave strategies whose long term behaviors are analysed in Section~\ref{sec:model}.  

Clusterized versions of structured-based overlays are such that  clusters of nodes substitute  nodes at the vertices of the graph.  Nodes are  uniquely identified with some $m$-bit string
randomly chosen from an ID-space. Identifiers (IDs)  are
derived by using standard collision-resistant one-way hash
functions (e.g.,~\cite{MD5}). Each graph vertex is composed of a set of
nodes  self-organised within a cluster according to some distance metrics (e.g., logical or geographical). Clusters in the system are uniquely labelled.  Size of each cluster is  lower  (resp. upper) bounded. The lower bound,  named $S_{min}$  in the following, usually satisfies some constraint based on the assumed failure model. For instance $S_{min} \geq4$  allows Byzantine tolerant agreement protocols to be run among these $S_{min}$ nodes~\cite{Lamport82}. The upper bound, that we call $S_{max}$, is typically in  $\mathcal O(log N)$, where  $N$ is the current number of nodes in the system, to meet  scalability requirements. 
When a cluster size reaches these bounds, cluster-based overlays react by
respectively splitting that cluster into two smallest clusters or by
merging it with its closest cluster neighbours. Finally for most of the cluster-based overlays, operations (\texttt{join}, \texttt{leave}, \texttt{merge}, and \texttt{split}) are poly-logarithmic in the number of nodes in the system.

In the present work we assume that at cluster level  nodes 
are organised as core and spare
members.  
Members of the core set are primarily responsible for handling messages
routing and clusters operations. Management of the core set is such that its size is
maintained to constant $S_{min}$.  Spare members are the complement
number of nodes in the cluster. In contrast to core members, they are not
involved in any of the overlay operations. Rationale of this classification is two-fold: first it allows to introduce the
unpredictability required to deal with Byzantine attacks through a randomized core set generation algorithm.  Second it 
limits  the management overhead caused by the natural churn present
in typical overlay networks through the spare set management. 

Specifically we consider the following join and leave operations: 
 \begin{itemize}
  \item \texttt{join}(p): when a peer \texttt{joins} a cluster, it joins it as a spare member.
  \item \texttt{leave}(p): 
  When a peer $p$ \texttt{leaves} a cluster either $p$ belongs to the spare set or  to the core set. In the former case, core
  members simply update their spare view to reflect $p$'s departure,
   while in the latter case, the core view maintenance procedure  is triggered.  Two different
  maintenance policies are implemented. The first one, referred in the
  following as \emph{policy 1}, simply consists in
  replacing the left core member by one randomly chosen spare
  member. The second one, referred as  \emph{policy 2},
  consists in refreshing the whole core set by choosing $S_{min}$
  random peers within the cluster. 
  \end{itemize}
For space reasons we do not give any detail regarding the localization of a cluster nor its creation/split/merge process. None of these operations are necessary for the understanding of our work. The interested reader is invited to read their description in the original papers (e.g.~\cite{ABLR08,equus,robustchord2005}). 

\subsection{Implementing a limited nodes lifetime}

To implement limited nodes lifetime, we propose to proceed   as follows:  Peers identifiers are generated based on
certificates acquired at trustworthy Certification Authorities
(CAs). Identifiers (denoted IDs) are  generated as the result of applying
a hash function  to some of the fields of a
X.509~\cite{RFC2459}  certificate.
To enforce  all peers, including malicious ones,  leaving and rejoining  the
system from time to time, we add a {\em incarnation} number to the
fields that appear in the peer's certificate  that will be hashed
to generate the peer's ID. The incarnation number limits the lifetime
of IDs. The current incarnation $k$ of any peer is given by the
following expression $k = \lceil (CT - IVT) \rceil / IL$,
where $IVT$ is the initial validity time of the peer's certificate, $CT$
is the current time, and $IL$ is the length of the lifetime of each
peer's incarnation. Thus, the $k^{th}$ incarnation of a peer $p$ expires
when its local clock reads $IVT+k*IL$. At this time $p$ must rejoin the
system using its $(k+1)^{th}$ incarnation.
The $IVT$ is one of the fields in the peer's certificate and since
certificates are signed by the CA, it cannot be unnoticeably modified
by a malicious peer. Moreover, a certificate commonly contains the
public key of the certified entity. This way, messages exchanged by
the peers can be signed using this key, preventing malicious peers
from unnoticeably altering messages originated from other peers in the
system. Messages must contain the certificate of their issuer, so as  to
allow recipients to validate them. Therefore, at any time, any peer
can check the validity of the ID of any other peers in the system, by
simply calculating the current incarnation of the other peer and
generating the corresponding ID. If some peer detects that the ID of
one of its neighbours is not valid then it cuts its connection with it.
Note that because clocks are loosely synchronised, it is
possible that a correct peer is still using its ID for incarnation $k$
when other correct peers would expect it to be in incarnation
$k+1$. To mitigate this problem, we assume that any correct peer may
have two subsequent valid incarnation numbers, for a fixed grace
window $GW$ of time that encompasses the expiration time of an
incarnation number ($GW$ is the maximum deviation of the clocks of any
two correct peers). More precisely, at any time $t$, both incarnation $k$
and $k'$ are valid, where:
$k = \lceil (t - GW/2 - IVT) \rceil / IL$, and
$k' = \lceil (t + GW/2 - IVT) \rceil / IL$.
Notice that this means that although at any time $t$ each peer $p$ has a
single incarnation number that it uses to define its current ID, other
peers calculate two possible incarnation numbers for $p$. These are 
frequently equal, but may differ  when $p$'s local time  is
close to the expiration time of its current/last incarnation.

\section{Modelling the adversarial strategy as a game}
\label{sec:model}

In this section, we investigate the previously described policies (policy 1 and 2). We model 
adversarial behavior by focusing on specific games. Both games intend
to prevent the adversary from elaborating deterministic strategies to
win. These games are played in the following context.  There is a
potentially infinite number of balls in a bag, with a proportion $\mu$
of red balls and a proportion $1-\mu$ of white balls, $\mu$ being a
constant in $(0,1)$.  White (resp. red) balls are
indistinguishable. Red balls are owned by the adversary. In addition
to the bag, there are two urns, named $\mathcal C$ and $\mathcal
S$. Initially, $c+s$ balls are drawn from the bag such that $c$ of
them are thrown into urn $\mathcal C$, and the other $s$ ones are
thrown into urn $\mathcal S$.  We denote by $C_r$ (resp. $S_r$) the
number of red balls in $\mathcal C$ (resp. $\mathcal S$). It is easily
checked that $C_r$ and $S_r$ are independent and have a binomial distribution, 
i.e. for
$x=0,\ldots,c$ and $y=0,\ldots,s$, we have
\begin{equation} \label{init}
\begin{array}{rcl}
\Pr\{C_r = x,S_r = y\} &=& \Pr\{C_r = x\}\Pr\{S_r = y\}\\
                    &   = & {c \choose x} \mu^x (1-\mu)^{c-x}
                         {s \choose y} \mu^y (1-\mu)^{s-y}.
                         \end{array}
\end{equation}
This joint distribution represents the initial distribution of the process 
detailed below.
Each  game is a succession of rounds $r_1,r_2, \ldots$ during which the game 
rule 
described in Figure~\ref{fig:game} is applied. 
Rules are oblivious to the colour of the balls, that is, 
they cannot distinguish between the white and the red balls. 

\begin{figure}\footnotesize
\centering

\hrule
\begin{minipage}[t]{6cm}

\begin{tabbing}
\renewcommand{\baselinestretch}{2.5}
 xx\=xx\=xx\=xx\=xx\=xx\=xx\=xx\=xx\=xx\=xxx\=xxx\=xxx\=\kill
 /* \textbf{First  game  }*/\\
/* \emph{stage 1} */\\
\>draw  ball $b_0$ from  $\mathcal C \cup \mathcal S$ \\ 
/* \emph{stage 2} */\ \\
\> \textbf{if} $b_0$ was in $\mathcal S$ \textbf{then} \\
 \>\> throw $b_0$ into the bag\\
 \>\> draw ball $b_2$ from the bag\\
 \>\> throw it into $\mathcal S$\\
\> \textbf{else}\\
 \>\> throw $b_0$ into the bag\\
\>\> draw ball $b_1$ from $\mathcal S$ \\
\>\> throw it into $\mathcal C$\\
\> \> draw ball $b_2$ from the bag \\
\>\> throw it into $\mathcal S$ 
\end{tabbing}
\end{minipage} \hspace{2cm}
\begin{minipage}[t]{6cm}
\begin{tabbing}
\renewcommand{\baselinestretch}{2.5}
 x\=xx\=xx\=xx\=xx\=xx\=xx\=xx\=xx\=xx\=xxx\=xxx\=xxx\=\kill
 /* \textbf{Second game  }*/\\
/* \emph{stage 1} */\\
\> draw  ball $b_0$ from  $\mathcal C \cup \mathcal S$ \\ 
/* \emph{stage 2} */\ \\
\>  \textbf{if} $b_0$ was in $\mathcal S$ \textbf{then} \\
 \>\> throw $b_0$ into the bag\\
 \>\> draw ball $b_2$ from the bag \\
 \>\> throw it into $\mathcal S$\\
\>  \textbf{else}\\
 \>\> throw $b_0$ into the bag\\
\>\> draw $c$ balls from $\mathcal S \cup \mathcal C$\\
\>\> throw these $c$ balls into $\mathcal C$\\
\>\> draw  one ball $b_2$ from the bag \\
\>\> throw it in $\mathcal S$
\end{tabbing}
\end{minipage}
 \hrule
\caption{\small Rule of the first and second game.}
\label{fig:game}
\end{figure}

The goal of the adversary is to get a quorum $Q$ of red balls in both urns $\mathcal C$ and $\mathcal S$
 so that  the number of red balls  in $\mathcal C$ 
is bound to continuously exceed   $\lfloor (c-1)/3 \rfloor$.  An intuition of why
having more than $\lfloor (c-1)/3 \rfloor$ red balls in urn
$\mathcal C$ is necessary for polluting it  is related to agreement
problems in distributed systems in presence of Byzantine processes.
 The value of quorum $Q$ is derived in the sequel.  
 The adversary may at any time inspect both urns and bag to elaborate
adversarial strategies to win the game. In particular it may not
follow the rule of the games by preventing its red balls from being
extracted from both urns. Specifically, at stage 1 of both games, if
the drawn ball $b_0$ is red then the adversary puts back the ball into
the urn from which it has been drawn. Stage 2 is not applied, and a
new round is triggered. Clearly this strategy ensures that the number
of red balls in $\mathcal C \cup \mathcal S$ is monotonically non
decreasing.

We model the effects of these rounds using a homogeneous Markov chain
denoted by $X = \{X_n, n \geq 0\}$ representing the evolution of
the number of red balls in both urns $\mathcal C$ and $\mathcal S$.
More formally, the state space $S$ of $X$ is defined by
$S = \{(x,y) \mid 0 \leq x \leq c, \; 0 \leq y \leq s\},$
and, for $n \geq 1$, the event $X_n =(x,y)$ means that, after the
$n$-th transition or $n$-th round, the number of red balls in urn
$\mathcal C$ is equal to $x$ and the number of red balls in urn
$\mathcal S$ is equal to $y$.  The transition probability matrix $P$ of
$X$ depends on the rule of the given game and on the adversarial
behaviours. This matrix is detailed in each of the following
subsections. In all the cases, the initial state $X_0$ is given by
$X_0 = (C_r,S_r)$ and its probability distribution is denoted by the
row vector $\alpha$ which is given by relation (\ref{init}), i.e.
$\alpha(x,y) = \Pr\{X_0=(x,y)\} = \Pr\{C_r = x,S_r = y\}.$

We define a state as \emph{polluted} if in that state urn $\mathcal C$
contains more than $\lfloor (c-1)/3 \rfloor$ balls. In the following, we denote by $c'$ the value  $\lfloor (c-1)/3 \rfloor$. Conversely,
a state that is not polluted is said \emph{safe}. 
The subset of safe states, denoted by $A$, is defined as:
$A  = \{(x,y) \mid 0 \leq x \leq c', \; 0 \leq y \leq s\},$
while the set of polluted states, denoted by 
 $B$, is  the subset $S-A$, i.e.
$B  = \{(x,y) \mid c'+1 \leq x \leq c, \; 0 \leq y \leq s\}.$
We partition matrix $P$ in a manner conformant to the decomposition of 
$S = A \cup B$, by writing
$$P = \left(\begin{array}{cc}
         P_A & P_{AB} \\
         P_{BA} & P_B \\
            \end{array}\right),$$
where $P_A$ (resp. $P_B$) is the sub-matrix of dimension $|A|\times |A|$
(resp. $|B|\times |B|$),
containing the transitions between states of $A$ (resp. $B$). 
In the same way, $P_{AB}$ (resp. $P_{BA}$) is the sub-matrix of dimension 
$|A|\times |B|$ (resp. $|B|\times |A|$),
containing the transitions from states of $A$ (resp. $B$) to states of $B$ 
(resp. $A$). We also partition the initial probability distribution $\alpha$ according
to the decomposition $S = A \cup B$, by writing $\alpha = (\alpha_A \;\; \alpha_B),$ 
where sub-vector $\alpha_A$ (resp. $\alpha_B$) contains the initial 
probabilities of states of $A$ (resp. $B$).

\subsection{First  game }
\label{sec:game1-active}

Regarding the first game, computation of the probabilities of the  transition matrix is illustrated in Figure~\ref{fig:tree}. In this tree, each edge is labelled by a probability and its corresponding event following the rule of the game (see Figure~\ref{fig:game}). This figure can be interpreted as follows:
  At round $r$, $r\geq 1$, starting from state  $(x;y)$ (root of the tree) the Markov chain can transit to four different states, namely $(x;y)$, $(x;y+1)$, 
  $(x+1;y)$, and $(x+1;y+1)$ (leaves of the tree). The probability associated to each one of these transitions is  obtained by summing the products of the probabilities discovered along each path starting from the root to the leaf corresponding to the target state. 

%
\tikzstyle{level 1}=[level distance=2cm, sibling distance=8cm]
\tikzstyle{level 2}=[level distance=2cm, sibling distance=3cm]
\tikzstyle{level 3}=[level distance=2cm, sibling distance=4.5cm]
\tikzstyle{level 4}=[level distance=2cm, sibling distance=2.cm]
%
\tikzstyle{bag} = [text width=3em, text centered]
\tikzstyle{end} = [circle, minimum width=3pt,fill, inner sep=0pt]
%
\begin{figure}\footnotesize
\centering
\begin{tikzpicture}[->,grow=down,scale=0.6]
\node[bag] {\tiny $(x;y)$}
    child {
        node[bag] {}    
            child {
                node[end, label=below:{\tiny $(x;y)$}]{}
                edge from parent
                node[left]{\tiny ($b_0$ is red) $\frac{x}{c}$\,}
            }
            child {
                node[bag]{}
child{
node[bag]{}
child{
node[end, label=below:{\tiny $(x+1;y-1)$}]{}
edge from parent
node[left]{\tiny ($b_2$ is white) {$1-\mu$}}
} 
child{ 
node[end, label=below:{\tiny $(x+1;y)$}]{}
edge from parent
node[right]{\tiny{$\mu$} \, ($b_2$ is red)}
}
edge from parent 
                            node[left] {\tiny    ($b_1$ is red) $\frac{y}{s}$\,}
                            }
child{
node[bag]{}
child{
node[end, label=below:{\tiny $(x;y+1)$}]{}
edge from parent
node[left]{\tiny $\mu$} 
} 
child{ 
node[end, label=below:{\tiny $(x;y)$}]{}
edge from parent
node[right]{\tiny $1-\mu$ ($b_2$ is white)}
}
edge from parent 
                            node[right] {$\frac{s-y}{s}$ \tiny ($b_1$ is white)}
                            }
edge from parent 
                   node[right] {$\frac{c-x}{c}$ \tiny ($b_0$ is white)} 
                   }
                  edge from parent 
                   node[left] {\tiny ($b_0 \in \mathcal C$)  \,$\frac{c}{c+s}$ \,}  
}
%
child {
node[bag] {}    
             child {
                node[end, label=below:{\tiny $(x;y)$}]{}
                edge from parent
                node[left]{\tiny ($b_0$ is red) $\frac{y}{s}$\,}
              }
             child {
                 node[bag]{}
                  child{
                  node[bag]{}
                  edge from parent
                   child{
                  node[end, label=below:{\tiny $(x;y)$}]{}
                  edge from parent
                  node[left]{\tiny $1-\mu$}
                  }
child{
node[end, label=below:{\tiny $(x;y+1)$}]{}
edge from parent
node[right]{\tiny $\mu$ ($b_2$ is red)}
}
node[right]{\tiny 1}
}
edge from parent 
                   node[right] {$\frac{s-y}{s}$ \tiny ($b_0$ is white)}
} 
edge from parent 
         node[right] {\,\,\,\,\,$\frac{s}{s+c}$ \tiny ($b_0 \in \mathcal S$)} 
}; 
\end{tikzpicture}
\caption{\small{Transition diagram for the computation of the 
transition probability matrix $P$ for the first game.}}
\label{fig:tree}
\end{figure}
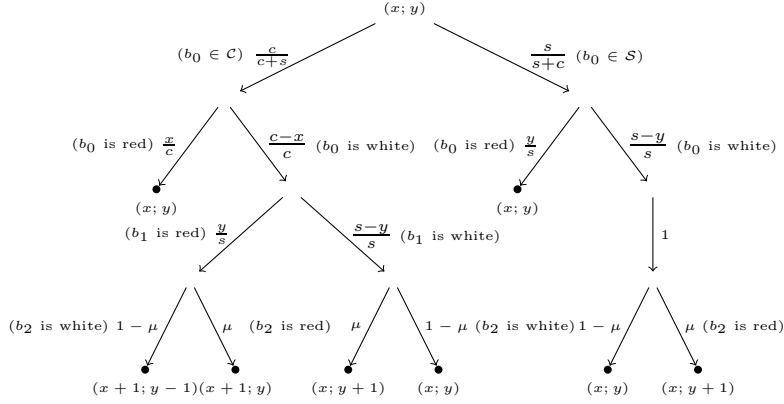

We can easily derive the transition
probability matrix $P$ of the Markov chain $X$ chain associated to this game.  
For all
$x \in \{0, \ldots, c\}$ and for all $y \in \{0, \ldots, s\}$, we have
\footnotesize
\begin{eqnarray*}
p_{(x,y),(x,y)}& = &\left(\frac{c}{c+s}\right) \left( \frac{x}{c} + 
\left(\frac{c-x}{c}\right)\left(\frac{s-y}{s}\right)(1-\mu)\right) + 
\left(\frac{s}{s+c}\right)\left(\frac{y}{s}\mu + 1-\mu\right) \\
p_{(x,y),(x,y+1)}& = &\left( \left(\frac{c}{c+s}\right)
\left( \frac{c-x}{c}\right) + \left(\frac{s}{s+c}\right) \right)
\left(\frac{s-y}{s}\right)\mu \;\;\; \mbox{ for }  y \leq s-1 \\
p_{(x,y),(x+1,y-1)}& = &\left(\frac{c}{c+s}\right)
\left( \frac{c-x}{c}\right) \frac{y}{s} (1-\mu) \;\;\; \mbox{ for } 
x \leq c-1 \mbox{ and } y \geq 1 \\
p_{(x,y),(x+1,y)}& = &\left(\frac{c}{c+s}\right)
\left( \frac{c-x}{c}\right) \frac{y}{s}\mu \;\;\; \mbox{ for } x \leq c-1.
\end{eqnarray*}
\normalsize
In all other cases, transition probabilities are null. 

Clearly, the adversary wins the game when the process $X$ reaches the
subset of states $B$ from which it cannot exit. Thus quorum $Q= \{(x,y) \mid (x,y) \in B\}.$ with $B$ the set of polluted states. By the rule of the game,
one can never escape from these states to switch to safe states since
the number of red balls in $\mathcal C$ is non decreasing. Thus there
is a finite random time $T$ after which the process $X$ is absorbed within
$B$. Thus we have $P_{BA}=0$. The Markov chain $X$ is reducible and
the states of $A$ are transient, which means that matrix $I-P_A$ is
invertible, where $I$ is the identity matrix of the right dimension
which is $|A|$ here.  Specifically $T$, the time needed to reach
subset $B$, is defined as $T = \inf\{n \geq 0 \mid X_n \in B\}.$
The cumulative 
distribution function of $T$ is easily derived as
\begin{equation}
\Pr\{T \leq k\} = 1 - \alpha_A (P_A)^{k}\nbu,
\end{equation}
where $\nbu$ is the column vector of the right dimension with all
components equal to $1$. The expectation of $T$ is given by
\begin{equation}
E(T) = \alpha_A (I - P_A)^{-1}\nbu,
\end{equation}
 


\subsection{Second game }
\label{sec:game2-active}

By proceeding similarly as above, we can derive the following transitions
of process $X$ associated to the second game. Briefly, when the game starts in state 
$(x,y)$ at round $r$, it remains in
state $(x,y)$ during the round  if either ball $b_0$ is
red or  $b_0$ is white, and has been drawn from $\mathcal S$, and $b_2$
is white. It changes to state $(x,y+1)$ if $b_0$ is white, it has been
drawn from $\mathcal S$, and $b_2$ is red. Finally the game switches to
state $(k, x+y-k+\ell)$, where $k$ is an integer $k=0,\ldots,
c'$ and $\ell=0$ or $1$ if $b_0$ is white, it has been drawn from
$\mathcal C$, and the renewal process leads to the choice of $k$ red
balls. 
 For all $x \in \{0, \ldots,c\}$ and $y \in \{0, \ldots, s\}$, we have
\footnotesize
\begin{eqnarray*}
p_{(x,y),(x,y)}&=&\left(\frac{c}{c+s}\right) \left(\frac{x}{c} \mu 
q(x,x+y-1) + \frac{c-x}{c} (1-\mu) q(x,x+y) \right)+\\ 
&& \left(\frac{s}{c+s}\right)\left(\frac{s-y}{s}\left(1-\mu\right) + 
\frac{y}{s}\mu\right) \\
p_{(x,y),(x,y+1)}&=&\left(\frac{c}{c+s}\right) \frac{x}{c} \mu q(x,x+y)   
+ \left(\frac{s}{c+s}\right)\left(\frac{s-x}{c} \mu \right) \mbox{ for }
y \leq s-1 \\
p_{(x,y),(x,y-1)}&= &\left(\frac{c}{c+s}\right)\frac{x}{c} (1-\mu) 
q(x,x+y-1) + \left(\frac{s}{c+s}\right)\left(\frac{s-x}{c} \mu \right) 
\mbox{ for }  y \geq 1 \\
p_{(x,y),(k,x+y-k)} &= &\left(\frac{c}{c+s}\right)
\left(\frac{c-x}{c}\right)\left(1-\mu\right)q(k,x+y) \\
&& \mbox{ for }  \max(0,x+y-s) \leq k \leq \min(c,x+y) 
\mbox{ and } k \neq x \\
p_{(x,y),(k,x+y-k+1)} &= &\left(\frac{c}{c+s}\right)
\left(\frac{c-x}{c}\right)\mu q(k,x+y) \\
&&\mbox{ for } \max(0,x+y+1-s) \leq k \leq \min(c,x+y+1) 
\mbox{ and } k \neq x
\end{eqnarray*}
\normalsize
where 
\footnotesize
$$q(x,x+y)=\frac{\DS{{x+y \choose x}{c+s-1-(x+y) \choose c-x}}}
                {\DS{{c+s-1 \choose c}}}$$
\normalsize
is the probability of getting $x$ red balls when $c$ balls are drawn, without 
replacement, in an urn containing $x+y$ red balls and $c+s-1 - (x+y)$ white 
balls, referred to as the hypergeometric distribution.
In all other cases, transition probabilities are null. 

In contrast to the first game, this game alternates between safe and polluted
states. After a random number of these alternations
the process ends by entering a set of closed polluted  states. Indeed, by the rule of the game, one can escape finitely
often from polluted state $(x;y)$ to switch back to a safe state as long as 
$(x;y)$ satisfies $c'+1 \leq x+y \leq s+c'$ (there are still sufficiently many white 
 balls in both $\mathcal C$ and  $\mathcal S$ so as to successfully withdrawing
$c$ balls such that $\mathcal C$ can be reverted to a safe state). However,
there is a time $T_D$ when  state $(x;y)$, with $x+y \geq s+c'+1$,
is entered. From $T_D$ onwards, going back to safe states
is impossible. Thus at time $T_D$ the adversary wins the game. Hence an interesting metrics to be evaluated is the total time spent by the
process in safe states before being definitely absorbed in polluted states.

Formally, we need to decompose the set $B$ of polluted states  into
two subsets $C$ and $D$
defined by
$C = \{(x;y) \mid c'+1 \leq x+y \leq s+c', \; c'+1 \leq x \leq c, \; 0 \leq y \leq s\},$
and
$D = \{(x;y) \mid x+y \geq s+c'+1, \; 0 \leq y \leq s\}.$
Subsets $A$ and $C$ are transient and subset $D$ is a closed subset.
We partition matrix $P$ and initial probability vector $\alpha$
following the decomposition of
$S = A \cup C \cup D$, by writing
$$P = \left(\begin{array}{ccc}
         P_A & P_{AC} & 0 \\
         P_{CA} & P_C & P_{CD} \\
           0    &  0  & P_D \\
            \end{array}\right)
\ \mbox{and} \ 
\alpha = (\alpha_A \;\; \alpha_C \;\; \alpha_D).$$
\noindent
Figure~\ref{fig:aggregated} illustrates the  states  partition of the process $X$. 
\begin{figure}
\centering
\begin{tikzpicture}[->,>=stealth',shorten >=1pt,auto,node distance=2cm,
                    semithick]

  \node[state]	(A) {A};
  \node[state]  (C) [right of=A] {C};
  \node[state]  (D) [right of=C] {D};

  \path (A) edge [loop above] node {} (A)
            edge              node {} (C)
        (C) edge [loop above] node {} (C)
            edge              node {} (A)
            edge              node {} (D)
        (D) edge [loop above]	node {} (D);
\end{tikzpicture}
\caption{\small An aggregated  view of the Markov chain associated to the second game. Safe states are represented by A, and polluted states by C and D.}
\label{fig:aggregated}
\end{figure}
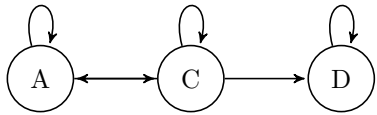

We are interested in the random variable $T_A$ which counts 
the total time spent in subset $A$ before reaching subset $D$.
Following the result obtained in~\cite{sericola}, 
we have, for every $k \geq 0$,
\begin{equation}
\Pr\{T_A \leq k\} = 1 - vG^{k}\nbu,
\end{equation}
where 
$v = \alpha_A + \alpha_C (I-P_C)^{-1}P_{CA} \mbox{ and }
  G = P_A  + P_{AC}(I-P_C)^{-1}P_{CA}.$
The expected total time spent in $A$ is given by
\begin{equation}
E(T_A) = v (I-G)^{-1}\nbu.
\end{equation}

Figure~\ref{fig:expectation-game} compares the expectation of the time spent in safe states for both games. In accordance with the intuition, increasing the size of the urns augments the expected time spent in safe states of both games, i.e., $E(T)$ and $E(T_A)$, independently of the ratio of red balls in the bag. Similarly, for a given cluster size, increasing the ratio of red balls in the bag drastically decreases both $E(T)$ and $E(T_A)$. However surprisingly enough, increasing the level of randomness (game 2 vs. game 1) does not increase the resilience to the adversary behavior since the first game  always overpasses the second one in expectation. It is even more true when $\mathcal S$ size is large with respect to $\mathcal C$ one. The intuition behind this fact is as follows: when $\mathcal S$ size is equal to 1, both games are equivalent   as illustrated in Figure~\ref{fig:expectation-game} for $s=1$. Now,  consider the case where the size of $\mathcal S$ is large with respect to $\mathcal C$ one. First of all, note that  the probability to draw a ball from  $\mathcal S$ tends to $1$, and because the adversary never withdraw its red balls from any urns,  the ratio of red balls within   $\mathcal S$ is monotonically non decreasing. Hence, the ratio of red balls in $\mathcal S$ tends also to 1. With small probability, a ball from $\mathcal C$ is drawn. In the first game   it is replaced with high probability by a red ball drawn from $\mathcal S$. Hence to reach a polluted state, at least $c'$ white balls have to be replaced by red ones. While in the second game  with high probability, the renewal of $\mathcal C$ reaches a polluted state in a single step. From this crude reasoning we can derive that the ratio of $E(T)$ over $E(T_A)$ tends to $c'$. 

%

\begin{figure}[htbp]
\centering
\subfigure[]{\includegraphics[height=8.5cm,width=8cm,angle=270]{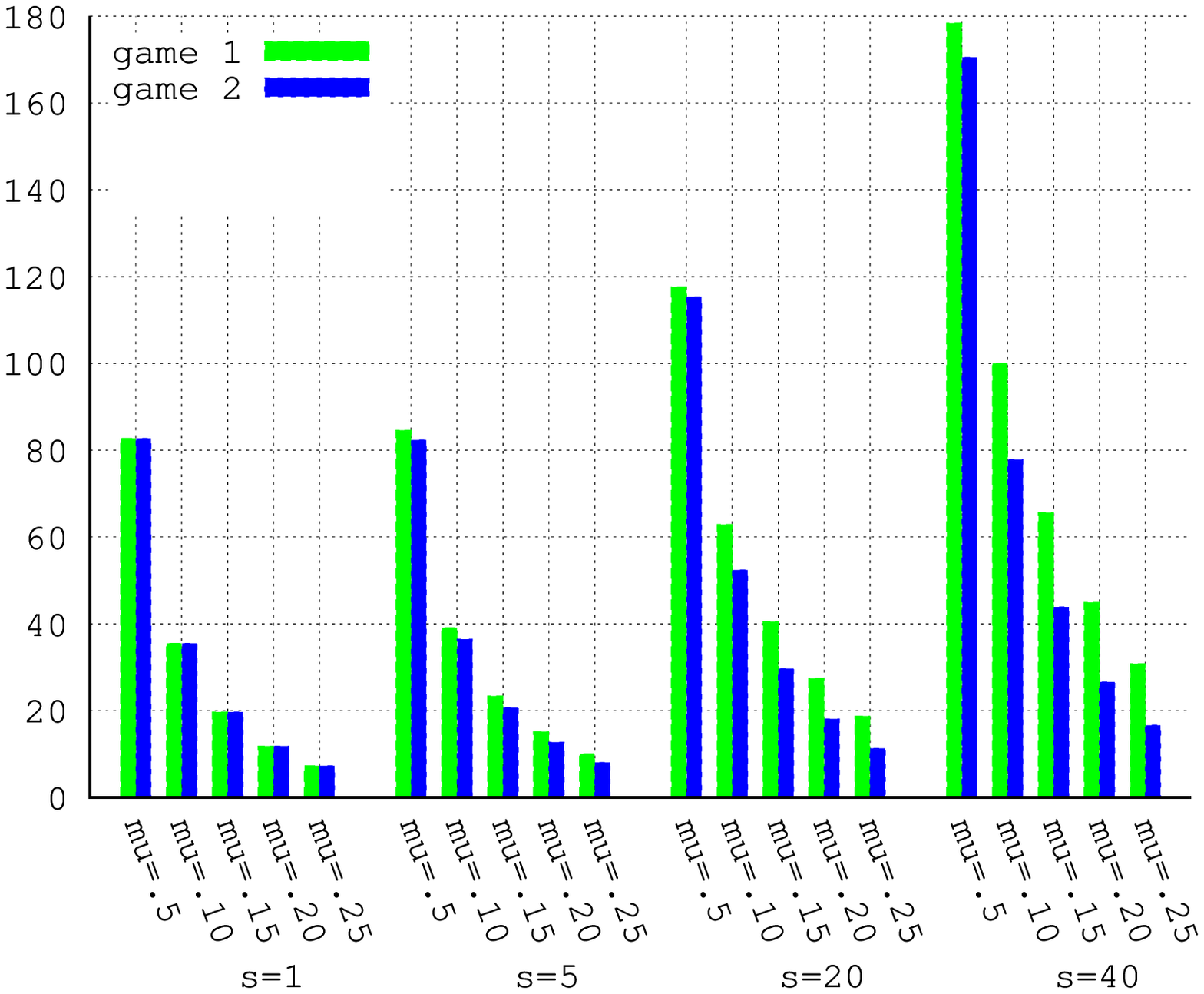}\label{fig:expectation-game}}
\subfigure[]{\includegraphics[height=8cm,width=7cm,angle=270]{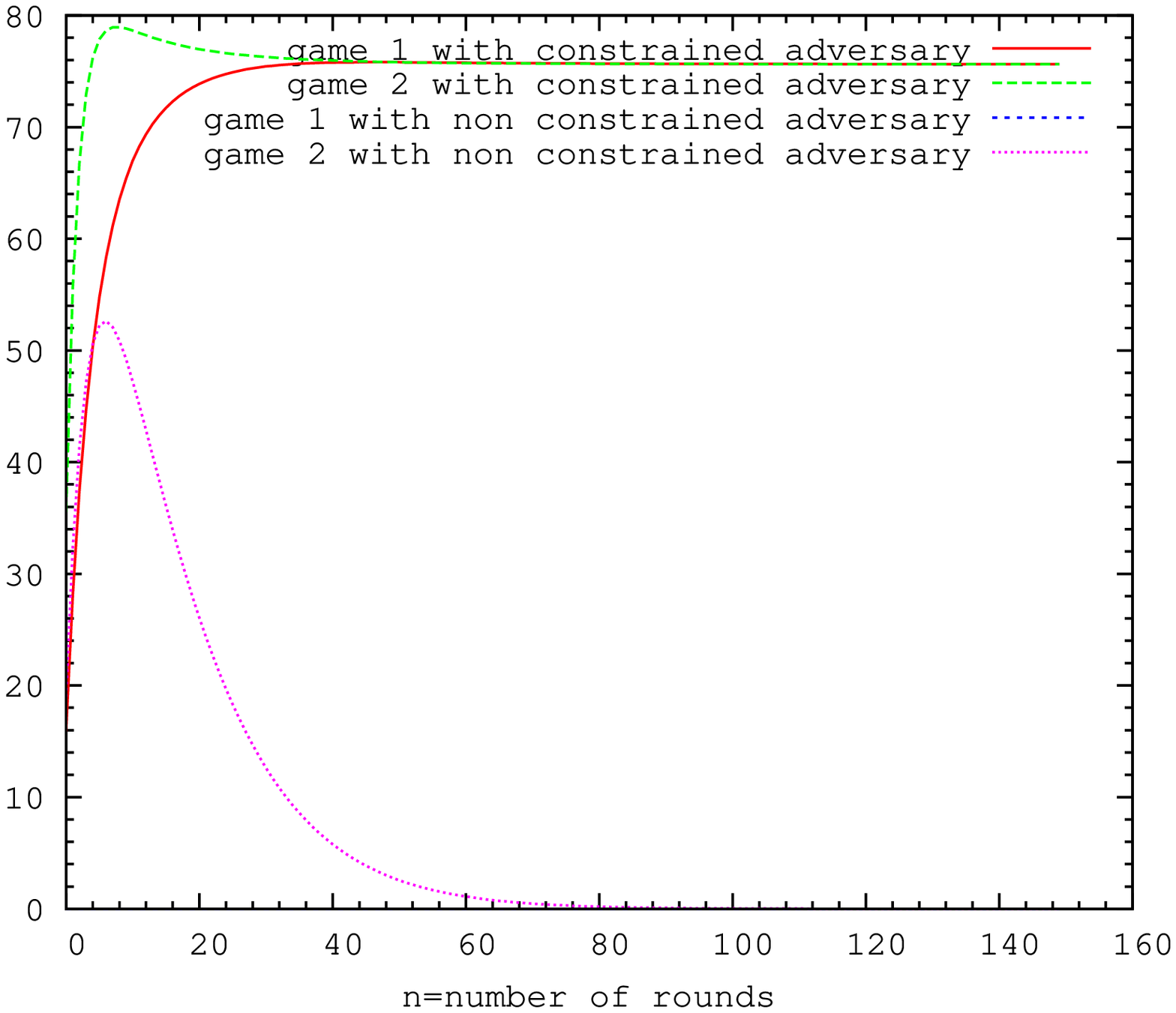}\label{fig:cluster}}
\caption{\small (a) Expectation of the number of rounds spent in safe states for games 1 and 2  function of $\mathcal S$ size and the ratio of malicious nodes $\mu$ as resp. given by relations~(\ref{eq:et}) and~(\ref{eq:eta}). (b) Mean number of safe clusters $E(N_n)$ (relation~(\ref{eq:enn})) in function of the rounds number $n$ for both games and both kind of adversaries. There are l=100 clusters, and the ratio of red balls in the bag is equal to .25 and $c=7$. Note that the initial number of safe clusters is equal to 16.}
\end{figure}

\section{Constraining the adversary}
\label{sec:constraint-model}

Our next step is  to evaluate the benefit of constraining the adversary by
limiting the sojourn time of its balls in both urns, so that randomness among red and white balls is continuously preserved. In the model we propose, 
we assume that the adversary cannot prevent red balls from being withdrawn for both urns. 

By proceeding as in Sections~\ref{sec:game1-active}
and~\ref{sec:game2-active}, we can derive the transition probability
matrix $P$ for both games. For all $x \in \{0, \ldots, c\}$ and
$y \in \{0, \ldots, s\}$, the entries of $P$ are given, for the first game, by
\footnotesize
\begin{eqnarray}
p_{(x,y),(x,y)}& = &\frac{xy + (c(s-y)-xs)(1-\mu)}{(c+s)s}) +
\frac{y \mu + (s-y)(1-\mu)}{c+s} \nonumber \\
p_{(x,y),(x,y-1)}& = &\frac{(x+s)y}{(c+s)s}(1-\mu) \mbox{ for } y \geq 1 
\nonumber \\
p_{(x,y),(x,y+1)}& = &\left(\frac{c-x+s}{c+s}\right)
\left(\frac{s-y}{s}\right)\mu \mbox{ for }  y \leq s-1 \nonumber \\
p_{(x,y),(x+1,y-1)}& = &\frac{(c-x)y}{(c+s)s}(1-\mu) 
\mbox{ for } x \leq c-1 \mbox{ and } y \geq 1 \label{V1} \\
p_{(x,y),(x+1,y)}& = &\frac{(c-x)y}{(c+s)s}\mu
\mbox{ for } x \leq c-1 \nonumber \\
p_{(x,y),(x-1,y)}& = &\frac{x(s-y)}{(c+s)s}(1-\mu) \mbox{ for } x \geq 1 
\nonumber \\
p_{(x,y),(x-1,y+1)}& = &\frac{x(s-y)}{(c+s)s}\mu
\mbox{ for }  x \geq 1 \mbox{ and } y \leq s-1. \nonumber 
\end{eqnarray}
\normalsize
In all other cases, transition probabilities are null. 
Similarly for second game , for all $x \in \{0, \ldots, c\}$ and 
$y \in \{0, \ldots, s\}$, we have
\footnotesize
\begin{eqnarray}
p_{(x,y),(x,y)}& = &\frac{x q(x,x+y-1)\mu + (c-x)q(x,x+y)(1-\mu)}{c+s} +
\frac{y \mu + (s-y) (1-\mu)}{c+s} \nonumber \\
p_{(x,y),(x,y-1)}& = &\frac{x}{c+s}q(x,x+y-1)(1-\mu) + \frac{y}{c+s}(1-\mu) 
\mbox{ for }  y \geq 1 \nonumber \\
p_{(x,y),(x,y+1)}& = &\frac{c-x}{c+s}q(x,x+y)\mu + \frac{s-y}{c+s}\mu 
\mbox{ for }  y \leq s-1 \nonumber \\
p_{(x,y),(k,x+y-k-1)}& = &\frac{x}{c+s}q(k,x+y-1)(1-\mu) \label{V2} \\ 
&&\mbox{ for } \max(0,x+y-1-s) \leq k \leq \min(c,x+y-1) \mbox{ and } k \neq x
\nonumber \\
p_{(x,y),(k,x+y-k)}& = &\frac{x}{c+s}q(k,x+y-1)\mu
+ \frac{c-x}{c+s}q{(k,x+y)}(1 - \mu) \nonumber \\
&&\mbox{ for } \max(0,x+y-s) \leq k \leq \min(c,x+y-1) \mbox{ and } k \neq x
\nonumber \\
p_{(x,y),(k,x+y-k+1)}& = &\frac{c-x}{c+s}q{(k,x+y)}\mu \nonumber \\
&&\mbox{ for } \max(0,x+y+1-s) \leq k \leq \min(c,x+y) \mbox{ and } k \neq x,
\nonumber 
\end{eqnarray}
\normalsize
where we set $q(u,v) = 0$ when $u > v$. In all other cases, transition probabilities are null. 

It is not difficult to see that none of the games exhibit an absorbing
class of states (i.e., both games never ends). 
We have $P_{BA} \neq 0$ and
the process $X$ is irreducible and aperiodic since at least
one  state has a transition to itself. 
The distribution of
the time $T$ needed to reach subset $B$ is given, for every $k \geq 0$, by
\begin{equation} \label{hitting}
\Pr\{T \leq k\} = 1 - \alpha_A (P_A)^{k}\nbu.
\end{equation}

We denote by $\pi$ the stationary distribution of the Markov chain $X$.
The row vector $\pi$ is thus the solution to the linear system
$$\pi = \pi P \mbox{ and } \pi \nbu = 1.$$
As we did for row vector $\alpha$, we partition  $\pi$ according
to the decomposition $S = A \cup B$, by writing $\pi = (\pi_A \;\; \pi_B),$ 
where sub-vector $\pi_A$ (resp. $\pi_B$) contains the stationary 
probabilities of states of $A$ (resp. $B$).

\begin{theorem} \label{stat}
For both games 1 and 2, the stationary distribution $\pi$ is equal to 
$\alpha$, i.e.
for all $x=0,\ldots,c$ and $y=0,\ldots,s$, we have
$$\lim_{n \longrightarrow \infty} \Pr\{X_n = (x,y)\} = \alpha(x,y),$$
which is given by relation (\ref{init}).
\end{theorem}
\begin{proof}
For space reasons, we omit the proof of the theorem. The interested reader is invited to read it in the Appendix. 
\cqfd
\end{proof}
Theorem~\ref{stat}  is interesting in two aspects. First it shows  that the stationary distribution $\pi$ is exactly the same for both games, and second, that this distribution is equal to the initial distribution $\alpha$. At a  first glance, we could guess that  this phenomenon  is due to the fact that the Markov chain $X$ is the tensor product of two independent Markov chains, representing respectively  the evolution of the red balls in $\mathcal C$ and $\mathcal S$. Although this is clearly not the case as the behavior of red balls in $\mathcal C$ depends on the behavior of red balls in $\mathcal S$. This holds for both games.

The stationary availability of the system defined by the long run
probability to be
in safe states is denoted by $P_{{\rm safe}}$ and is given by
$$P_{{\rm safe}} = \pi_A \nbu = 
\sum_{x=0}^{c'} {c \choose x} \mu^x (1-\mu)^{c-x}.$$
This probability can also be interpreted as the long run proportion of time
spent in safe states. Note that the stationary distribution does not depend on  the size of  $\mathcal S$.

Now let us consider that we have $\ell$ identical and independent Markov 
chains $X^{(1)}, \ldots, X^{(\ell)}$ on the same state space $S= A\cup \{S \setminus A\}$,
with initial probability distribution $\beta$ and transition 
probability matrix $P$. The probability distribution $\beta$ represents the state $(0;0)$, i.e., the safest state. Each Markov chain models  a particular cluster of
nodes and, for $n \geq 0$,  $N_n$ represents the number of safe clusters 
after the $n$-th round, i.e. the number of Markov chains being in subset $A$
after the $n$-th transition has been triggered, defined by
$$N_n = \sum_{j=1}^{\ell} 1_{\{X_n^{(j)} \in A\}}.$$
The $\ell$ Markov chains being identical and independent, $N_n$ has a binomial
distribution, that is, for $k=0,\ldots,\ell$, we have
\begin{eqnarray*}
\Pr\{N_n = k\} & = & {\ell \choose k} \left(\Pr\{X_n^{(1)}\in A\}\right)^k 
                     \left(1 - \Pr\{X_n^{(1)}\in A\}\right)^{\ell - k} \\
& = & {\ell \choose k} \left(\beta P^n \nbu_A\right)^k
                       \left(1 - \beta P^n \nbu_A\right)^{\ell - k}
\end{eqnarray*}
and 
$$E(N_n) = \ell \beta P^n \nbu_A,$$
where $\nbu_A$ is the column vector with the $i$-th entry equal to $1$ if 
$i \in A$ and equal to $0$ otherwise.
If $N$ denotes the stationary number of safe clusters, we have,
for $k=0,\ldots,\ell$,
$$
\begin{array}{rcll}
\Pr\{N = k\} &=& {\ell \choose k} \left(\pi_A\nbu\right)^k 
                 \left(1 - \pi_A\nbu\right)^{\ell - k} &\mbox{ for a constrained adversary } \\
                & = & 0 & \mbox{ for a non constrained adversary}
                \end{array}
                $$
and
$$
\begin{array}{rcll}
E(N) &= &\ell \pi_A\nbu &\mbox{ for a constrained adversary }\\
	& = & 0 & \mbox{ for a non constrained adversary}
	\end{array}$$

These results are illustrated in Figure~\ref{fig:cluster}. We can observe that with a constrained adversary, the ratio of safe clusters tends to the same limit for both games, whatever the amount of initially safe clusters (less than a 1/4), while with a non constrained adversary eventually  all the clusters  get polluted.

\section{Conclusion}
In this paper, we have proposed a mechanism that enables the enforcement of limited nodes lifetime compliant with DHT-based overlays specificities. We have  investigated several adversarial strategies. Our analysis  has demonstrated  that
 an adversary can easily  subvert a cluster-based overlay  by simply never triggering leave operations. We have  then shown  that  when  nodes have to regularly leave the system, eventually this one reaches a stationary regime where the ratio of malicious nodes is bounded. 
 
 For future work, we plan to implement this limited node lifetime mechanism in PeerCube to study its impact on the induced churn and its management overhead.  We are convinced that this additional churn will be efficiently amortised  thanks to the organisation of nodes in core and spare sets.
 
\bibliographystyle{plain}\small
\bibliography{references}
\newpage
\section*{Appendix}
\begin{theorem-repeat}{stat}
For both games 1 and 2, the stationary distribution $\pi$ is equal to 
$\alpha$, i.e.
for all $x=0,\ldots,c$ and $y=0,\ldots,s$, we have
$$\lim_{n \longrightarrow \infty} \Pr\{X_n = (x,y)\} = \alpha(x,y),$$
which is given by relation (\ref{init}).
\end{theorem-repeat}

\begin{proof}
For both games, the Markov chain $X$ is finite, irreducible and aperiodic
so the stationary distribution exists and is unique. It thus suffices to 
show that for both games we have $\alpha = \alpha P$, i.e. for all
$i \in \{0,\ldots,c\}$ and $j \in \{0,\ldots,s\}$, we have
$$(\alpha P)(i,j) = 
\sum_{u=0}^c \sum_{v=0}^s \alpha(u,v) p_{(u,v),(i,j)} = \alpha(i,j).$$
First of all, note that, from relation (\ref{init}), we have
$$\begin{array}{lcll}
\alpha(i,j+1) & = & \alpha(i,j)\frac{(s-j)\mu}{(j+1)(1-\mu)} 
& \mbox{ for } j \leq s-1, \\\\
\alpha(i,j-1) & = & \alpha(i,j)\frac{j(1-\mu)}{(s-j+1)\mu} 
& \mbox{ for } j \geq 1, \\\\
\alpha(i-1,j+1) & = & \alpha(i,j)\frac{i(s-j)}{(c-i+1)(j+1)} 
& \mbox{ for } i \geq 1 \mbox{ and } j \leq s-1, \\\\
\alpha(i-1,j) & = & \alpha(i,j)\frac{i(1-\mu)}{(c-i+1)\mu} 
& \mbox{ for } i \geq 1, \\\\
\alpha(i+1,j) & = & \alpha(i,j)\frac{(c-i)\mu}{(i+1)(1-\mu)} 
& \mbox{ for } i \leq c-1, \\\\
\alpha(i+1,j-1) & = &\alpha(i,j)\frac{(c-i)j}{(i+1)(s-j+1)} 
& \mbox{ for } i \leq c-1 \mbox{ and } j \geq 1.
\end{array}$$

For first game , the transition probability matrix $P$ is given by relations
(\ref{V1}). Using these relations and relations above, we obtain
for $i =1,\ldots,c-1$ and $j = 1,\ldots,s-1$,
\begin{eqnarray*}
(\alpha P)(i,j) & = & \alpha(i,j)p_{(i,j),(i,j)} + 
              \alpha(i,j+1)p_{(i,j+1),(i,j)} + \alpha(i,j-1)p_{(i,j-1),(i,j)} \\
                &   & + \; \alpha(i-1,j+1)p_{(i-1,j+1),(i,j)} + \alpha(i-1,j)p_{(i-1,j),(i,j)} \\
                &   & + \; \alpha(i+1,j)p_{(i+1,j),(i,j)} +  \alpha(i+1,j-1)p_{(i,j),(i,j)} \\ 
& = & \alpha(i,j) \left(\frac{ij\mu + (c-i)(s-j)(1-\mu)}{(c+s)s} + 
                        \frac{j\mu + (s-j)(1-\mu)}{c+s}\right. \\
&   & + \; \frac{\mu(s-j)i}{(c+s)s} + \frac{\mu(s-j)}{c+s}
+ \frac{(1-\mu)j(c-i)}{(c+s)s} + \frac{(1-\mu)j}{c+s} \\
&   & \left. + \; \frac{i(1-\mu)}{c+s} 
+ \frac{(c-i)\mu}{c+s}\right) \\
& = & \alpha(i,j).
\end{eqnarray*}
When $i=0$ or $i=c$ and $j=0$ or $j=s$ we obtain the same result more easily.

 For second game , the transition probability matrix $P$ is given by relations
(\ref{V2}). For $i=1,\ldots,c-1$ and $j=1,\ldots,s-1$, we have
\begin{eqnarray*}
(\alpha P)(i,j) & = & \alpha(i,j) \frac{j\mu + (s-j)(1-\mu)}{c+s}
                      + \alpha(i,j+1) \frac{(j+1)(1-\mu)}{c+s} \\
&  & + \; \alpha(i,j-1) \frac{(s-j+1)\mu}{c+s} + \sum_{(u,v) \in S_{i+j+1}} \alpha(u,v) \frac{u(1-\mu)}{c+s}q(i,i+j) \\
&  & + \sum_{(u,v) \in S_{i+j}} \alpha(u,v) \left(\frac{u\mu}{c+s}q(i,i+j-1) 
     + \frac{(c-u)(1-\mu)}{c+s}q(i,i+j)\right) \\
&  & + \sum_{(u,v) \in S_{i+j-1}} \alpha(u,v) \frac{(c-u)\mu}{c+s}q(i,i+j-1), 
\end{eqnarray*}
where $S_\ell$ is the set defined by 
$S_\ell = \{(u,v) \mid 0 \leq u \leq c, \; 0 \leq v \leq c \mbox{ and } 
u+v = \ell\}.$
Using the recurrence relations above on $\alpha$ and two variables changes
$u:=u+1$ and $u:=u-1$, we obtain
\begin{eqnarray*}
(\alpha P)(i,j) & = 
 & \alpha(i,j) \frac{s}{c+s} 
+ \sum_{(u,v) \in S_{i+j}} \alpha(u,v) \frac{(c-u)\mu}{c+s}q(i,i+j) \\
&  & + \sum_{(u,v) \in S_{i+j}} \alpha(u,v) \left(\frac{u\mu}{c+s}
q(i,i+j-1) + \frac{(c-u)(1-\mu)}{c+s}q(i,i+j)\right) \\
&  & + \sum_{(u,v) \in S_{i+j}} \alpha(u,v) \frac{u(1-\mu)}{c+s} q(i,i+j-1),
\end{eqnarray*}
which leads to
$$(\alpha P)(i,j) = \frac{\alpha(i,j)s}{c+s}
+ \sum_{(u,v) \in S_{i+j}} \alpha(u,v) \left(\frac{c-u}{c+s}q(i,i+j) 
+ \frac{u}{c+s} q(i,i+j-1)\right).$$
By definition of $q(i,i+j)$, we have
$$q(i,i+j-1) = q(i,i+j)\frac{j(c+s-(i+j))}{(i+j)(s-j)}$$
and by definition of $\alpha(u,v)$, we have
$$\sum_{(u,v) \in S_{i+j}} u \alpha(u,v) = {c+s \choose i+j} \mu^{i+j} 
(1-\mu)^{c+s-(i+j)}\frac{(i+j)c}{c+s},$$
and thus
$$\sum_{(u,v) \in S_{i+j}} (c-u) \alpha(u,v) = {c+s \choose i+j} \mu^{i+j} 
(1-\mu)^{c+s-(i+j)}\frac{c(c+s-(i+j))}{c+s}.$$
This leads to
$$(\alpha P)(i,j) = \frac{\alpha(i,j)s}{c+s} + 
\frac{\DS{{c+s \choose i+j} \mu^{i+j} (1-\mu)^{c+s-(i+j)} 
q(i,i+j) cs(c+s-(i+j))}}{(c+s)^2(s-j)}.$$
Again, by definition of $q(i,i+j)$, we have
$${c+s \choose i+j} \mu^{i+j} (1-\mu)^{c+s-(i+j)} q(i,i+j)
= \alpha(i,j)\frac{(c+s)(s-j)}{s(c+s-(i+j))},$$
which gives $(\alpha P)(i,j) = \frac{\alpha(i,j)s}{c+s} + \frac{\alpha(i,j)c}{c+s} =
\alpha(i,j).$
As for game 1, the result for frontier states is easier to derive.
\cqfd
\end{proof}

\end{document}